\documentclass[final,3p,times]{elsarticle}
\usepackage[utf8]{inputenc}
\usepackage{microtype} 
\usepackage{algorithm}
\usepackage{algpseudocode}
\usepackage {amsmath}
\usepackage {xspace}
\usepackage {comment}
\usepackage {fancybox}
\usepackage {ifthen}
\usepackage{cleveref}
\usepackage {amsfonts,amssymb,bbold}

\usepackage[most]{tcolorbox}
\usepackage{amssymb}
\usepackage{amsthm}
\usepackage{color}
\usepackage{mdframed}
\usepackage{lipsum}
\usepackage{efbox,graphicx}
\efboxsetup{linecolor=green,linewidth=10pt}
\usepackage[compatibility=false]{caption}
\usepackage{mwe}

\DeclareMathOperator{\colvis}{ColVis}
\DeclareMathOperator{\vorvis}{VorVis}
\DeclareMathOperator{\kvorvis}{VorVis_k}
\DeclareMathOperator{\vis}{Vis}
\newcommand{\vorviewshed}[3][\T]{\ensuremath{\mathcal{W}_{#1}(#2,#3)}\xspace}
\newcommand{\mkmcal}[1]{\ensuremath{\mathcal{#1}}\xspace}

\newcommand{\T}{\mkmcal{T}}
\newcommand{\A}{\mkmcal{P}}

\newtheorem{lemma}{Lemma}
\newtheorem{theorem}{Theorem}

\newtheorem{example}{Example}
\efboxsetup{linecolor=green,linewidth=10pt}

\begin{document}

\begin{frontmatter}






\title{On Voronoi visibility maps of 1.5D terrains\\ with multiple viewpoints\tnoteref{t1}}
\tnotetext[t1]{Supported by
the Czech Science Foundation, grant number GJ19-06792Y, and with institutional support RVO:67985807. V.K was also partially supported by Charles University project UNCE/SCI/004 and by the Czech Academy of Sciences (Praemium Academiae awarded to M. Palu\v{s}).} 

\author[1,2]{Vahideh Keikha}
\ead{keikha@cs.cas.cz}
\author[1,3]{Maria Saumell\corref{cor1}}
\ead{maria.saumell@fit.cvut.cz}
\cortext[cor1]{Corresponding author}
\address[1]{The Czech Academy of Sciences, Institute of Computer Science, Czech Republic.
}
\address[2]{Faculty of Mathematics and Physics, Department of Applied Mathematics, Charles University, Prague, Czech Republic.}
\address[3]{Department of Theoretical Computer Science, Faculty of Information Technology, Czech Technical University in Prague, Czech Republic.}

\begin{abstract}
Given an $n$-vertex 1.5D terrain $\T$ and a set $\A$ of $m<n$ viewpoints, the Voronoi visibility map $\vorvis(\T,\A)$ is a partitioning of $\T$ into regions such that each region is assigned to the closest (in Euclidean distance) visible viewpoint. The colored visibility map $\colvis(\T,\A)$ is a partitioning of $\T$ into  regions that have the same set of visible viewpoints. In this paper, we propose an algorithm to compute $\vorvis(\T,\A)$ that runs in $O(n+(m^2+k_c)\log n)$ time, where $k_c$ and $k_v$ denote the total complexity of $\colvis(\T,\A)$ and $\vorvis(\T,\A)$, respectively. This improves upon a previous algorithm for this problem. We also generalize our algorithm to higher order Voronoi visibility maps, and to Voronoi visibility maps with respect to other distances. Finally, we prove bounds relating $k_v$ to $k_c$, and we show an application of our algorithm to a problem on limited range of sight.
\end{abstract}




\end{frontmatter}


\section{Introduction} 
A 1.5D terrain $\T$ is an $x$-monotone polygonal chain of $n$ vertices in $\mathbb{R}^2$. Two points on $\T$ are \emph{visible} if the segment connecting them does not contain any
point strictly below $\T$.  

Visibility problems in terrains are fundamental in geographical information science and have many applications, such as placing fireguard or telecommunication towers~\cite{crsar-fpp-07}, identifying areas that are not visible from sensitive sites~\cite{m-cwpl-06}, or solving problems related to sensor networks~\cite{ymg-wsns-08}. 
Although 2.5D terrains are more interesting for modelling and forecasting, 1.5D terrains are easier
to visualize and to analyze. They give insights into the difficulties of 2.5D terrains in terrain analysis, and their proper understanding is seen as an essential step towards the ultimate goal of settling the 2.5D case. For this reason, visibility problems in 1.5D terrains have been intensively studied by the computational geometry community during the last 15 years.

In this paper, we focus on the variant where a set 
$\A$ of $m<n$ viewpoints are located on vertices of \T (we refer to the end of this section for a discussion on the assumption $m<n$). For each viewpoint $p\in \A$, the \emph{viewshed} of $p$ is the set of points of $\T$ that are visible from $p$  (see Fig.~\ref{fig:viewshed} for an example). Our goal is to efficiently extract information about the visibility of $\T$ with respect to $\A$. We continue the work initiated in~\cite{ter-vis2014}, where the following structures are introduced.

\begin{figure}[h]
    \centering
   \includegraphics[scale=1]{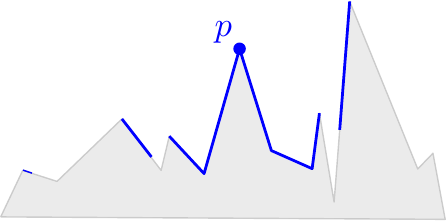}
    \caption{The viewshed of $p$.}
    \label{fig:viewshed}
\end{figure}



\begin{figure*}[h]
    \centering
   \includegraphics[scale=1]{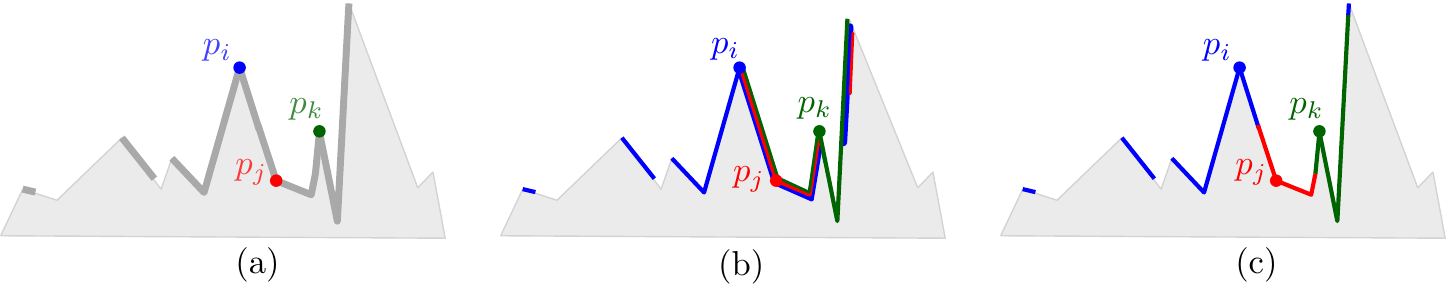}
    \caption{(a) $\vis(\T,\A)$ (the visible region is shown in gray). (b) $\colvis (\T,\A)$. (c)  $\vorvis (\T,\A)$.}
    \label{fig:examples-maps}
\end{figure*}

The {\em visibility map} $\vis(\T,\A)$ is a partitioning of $\T$ into a {\em visible} region (containing all portions of $\T$ that are visible by at least one element in $\A$) and an {\em invisible} region (containing the portions that are not visible by any element in $\A$). See Fig.~\ref{fig:examples-maps}a for an example. The visible region of the visibility map is equal to the union of the viewsheds of all viewpoints in $\A$. 

The {\em colored visibility map} $\colvis(\T,\A)$ is a partitioning of $\T$ into  regions that have the same set of visible viewpoints. See Fig.~\ref{fig:examples-maps}b for an example.

Finally, the {\em Voronoi visibility map} $\vorvis(\T,\A)$ is a partitioning of $\T$ into  regions that have the same closest visible viewpoint, where the distance used is the Euclidean distance (not the distance along the terrain). See Fig.~\ref{fig:examples-maps}c for an example.

 



Algorithms to compute these structures for both 1.5D and 2.5D terrains are proposed in~\cite{ter-vis2014}. The algorithm to obtain $\vorvis(\T,\A)$ of a 1.5D terrain runs in $O(n + (m^2 + k_c) \log n + k_v(m + \log n \log m))$  time, where $k_c$ and $k_v$ denote the total complexity of $\colvis(\T,\A)$ and $\vorvis(\T,\A)$, respectively. Both $k_c$ and $k_v$ have size $O(mn)$, and this bound is asymptotically tight~\cite{ter-vis2014}. The algorithm first computes $\colvis(\T,\A)$, and then it spends $\Theta(m)$ time to find each single region of $\vorvis(\T,\A)$. In this paper, we show that $\vorvis(\T,\A)$ can be extracted from $\colvis(\T,\A)$ in a 1.5D terrain 
much more efficiently, 
resulting in an $O(n+(m^2+k_c)\log n)$-time algorithm. We use an observation related to intersections of the terrain with bisectors of pairs of viewpoints that also allows us to prove a relationship between $k_c$ and $k_v$.

Let us point out that, apart from the mentioned output-sensitive algorithm for $\vorvis(\T,\A)$ of a 1.5D terrain, the authors of~\cite{ter-vis2014} also propose a divide-and-conquer algorithm running in $O(mn\log m)$ time, which is worst-case nearly optimal (recall that the maximum complexity of $\vorvis(\T,\A)$ is $\Theta(mn)$). Therefore, our new algorithm does not represent an improvement in the worst-case instances, but in instances where the original output-sensitive algorithm is faster than the divide-and-conquer one, and $k_vm$ is the dominant term in the running time. An example of such an instance is $m=\Theta(\sqrt{n})$, $k_c=\Theta(n^{3/4})$ and $k_v=\Theta(n^{3/4})$.

In this paper, we also provide generalizations of our algorithm to compute Voronoi visibility maps of higher order (that is, containing the information about the $k$ closest visible viewpoints, for some $k>1$), and Voronoi visibility maps with respect to two other distances: the Euclidean distance along the terrain and the link distance. All of these generalizations have the same running time as the original algorithm.

Finally, the new algorithm for $\vorvis(\T,\A)$ also allows us to solve efficiently a problem related to limited range of sight. These problems are motivated by the fact that, even though many visibility problems assume an infinite range of visibility, the intensity of light, signals and other phenomena modelled with viewpoints decreases over distance in realistic environments. 
In this spirit, the problem of illuminating a polygonal area with the minimum total energy was introduced by O'Rourke~\cite{ORourke}, and studied in~\cite{eisenbrand2008energy,ernestus2017algorithms}. 
We consider a related problem on terrains, namely, computing the minimum value $r^*$ such that, if the viewpoints can only see objects within distance $r^*$, the obtained visibility map is the same as $\vis(\T,\A)$. 
We show that this problem can also be solved in $O(n+(m^2+k_c)\log n)$ time.

\paragraph{Related Work} 
When there is only one viewpoint, computing the visibility map of a 1.5D terrain can be done in $O(n)$ time by converting the terrain into a simple polygon and applying the algorithm from~\cite{joe1987corrections}. 
One of the first results on the variant with more than one viewpoint is an $O((n+m) \log m)$ time algorithm to detect if there are any visible pairs of
viewpoints above a 1.5D terrain~\cite{ben2004computing}. Later, a systematic study
of $\vis(\T,\A)$, $\vorvis(\T,\A)$ and $\colvis(\T,\A)$ was carried out in~\cite{ter-vis2014} for both 1.5D and 2.5D terrains. A problem that is very related to the construction of $\vis(\T,\A)$ is that of computing the total visibility index of the terrain, that is, the number of viewpoints
that are visible from each of the viewpoints. This problem can be solved in $O(n\log^2 n)$ time~\cite{afshani2018efficient}.


The situation where the locations of the viewpoints are unknown has been thoroughly studied. 
It is well-known that computing the minimum number of viewpoints to keep a 1.5D terrain illuminated is NP-complete~\cite{friedrichs2015continuous,king2011terrain}, but the problem admits a PTAS~\cite{friedrichs2015continuous,friedrichs2014ptas,gibson2009approximation}. If the viewpoints are restricted to lie on a line, the same problem can be solved in linear time~\cite{daescu2019altitude}. 

\begin{figure*}[t]
    \centering
    \includegraphics[scale=1]{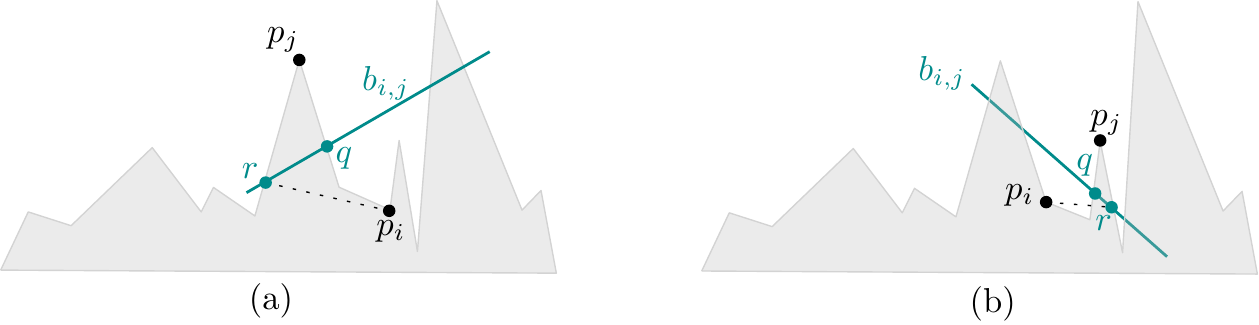}
    \caption{Illustration of Lemma~\ref{lem:key}: (a) $x(p_i)>x(p_j)$; (b) $x(p_i)<x(p_j)$. }
    \label{fig:bisectorij}
\end{figure*}

\paragraph{Assumptions} As in~\cite{ter-vis2014}, we assume that no three vertices of $\T$ are aligned. For the sake of simplicity, we also assume that no edge of $\T$ is contained in the bisector of two viewpoints in $\A$, and that no point on $\T$ is at the same distance from three or more viewpoints in $\A$.

As mentioned earlier, we restrict to the case where the viewpoints lie on terrain vertices; the same assumption is made in~\cite{ter-vis2014}, and it has the implication that $m \leq n$. Notice that no generality is lost because, if viewpoints are located in the interior of terrain edges, we can simply add vertices to the terrain and apply our algorithms. Furthermore, placing a superlinear number of viewpoints on the terrain does not seem to make much sense: If more than two viewpoints lie on the same edge, it is easy to see that the union of the viewsheds of the leftmost and rightmost viewpoints contains the viewshed of any other viewpoint on the edge. Therefore, all the intermediate viewpoints are somewhat irrelevant for visibility purposes. 

Finally let us mention that, in~\cite{ter-vis2014}, $k_v$ and $k_c$ do not only include the number of points of $\T$ that are on the boundary of two distinct regions of the respective diagrams, but also the total number of vertices of $\T$, that is, $n$. For the sake of consistency, we follow the same convention in this paper.

\section{Complexity of the Voronoi visibility map}

In~\cite{ter-vis2014}, it is stated that the complexity of $\vorvis(\T,\A)$ can be higher than, lower than,
or equal to that of $\colvis (\T,\A)$. In this section, we refine this statement. Recall that in both cases the complexity is $O(mn)$, and this bound is asymptotically tight~\cite{ter-vis2014}.

Let us introduce some terminology. The \emph{Voronoi viewshed}
\vorviewshed{p}{\A} of $p$ is the set of points in the viewshed of
$p$ that are closer to $p$ than to any other viewpoint that is visible from them. 

 Since we have assumed that no edge of $\T$ is contained in the bisector of two viewpoints, the shared boundary between two consecutive regions of $\vorvis (\T,\A)$ is always a single point of $\T$. We call such points \emph{event} points of $\vorvis (\T,\A)$. \emph{Event} points of $\colvis (\T,\A)$ are defined analogously, that is, as points on the boundary of two consecutive regions of the map.

We denote by $b_{i,j}$ the perpendicular bisector of two viewpoints $p_i,p_j$. Additionally, we denote by $q_{i,j}$ an event point of $\vorvis (\T,\A)$ such that a point infinitesimally to the left and right of $q_{i,j}$ belongs to $\vorviewshed{p_i}{\A}$ and $\vorviewshed{p_j}{\A}$, respectively (notice that an event $q_{i,j}$ is different from an event $q_{j,i}$). There are three (not mutually exclusive) possibilities: (i) $p_i$ becomes invisible at $q_{i,j}$\footnote{When we write that $p_i$ becomes invisible at $q_{i,j}$, we mean that it is visible immediately to the left of $q_{i,j}$ and invisible immediately to its right. We use the same rational when we write that $p_j$ becomes visible at $q_{i,j}$.}; (ii) $p_j$ becomes visible at $q_{i,j}$; (iii) $p_i$ and $p_j$ are visible at $q_{i,j}$, and $q_{i,j}$ is an intersection point between $b_{i,j}$ and $ \T$.


In the following lemma, we prove the key observation of this paper: Even though a bisector $b_{i,j}$ might intersect the terrain $\Theta(n)$ times, only two such intersections are relevant and might produce events of type (iii).

\begin{lemma} \label{lem:key}
Let $p_i\in \A$ be lower\footnote{We say that $p$ is \emph{lower} (respectively, \emph{higher}) than $q$ when it has a smaller (respectively, greater) $y$-coordinate than that of $q$.} than $p_j\in \A$. Let $q$ be an intersection point between $b_{i,j}$ and $ \T$ to the left (respectively, right) of $p_i$. Then any point to the left (respectively, right) of $q$ that is visible from $p_i$ is closer to $p_j$ than to $p_i$. Hence, there is no event $q_{i,j}$ or $q_{j,i}$ of type (iii) that lies to the left (respectively, right) of~$q$.
\end{lemma}

\begin{proof} Since $p_i$ is assumed to have a smaller $y$-coordinate than that of $p_j$, 
the region of the plane closer to $p_i$ than to $p_j$ is the one below $b_{i,j}$. But any point $r$ that is on $b_{i,j}$ or below it and to the left (respectively, right) of $q$ is not visible from $p_i$ because the line segment $\overline{p_ir}$ contains a
point (specifically, the point vertically aligned with $q$) that lies strictly below the terrain surface; see Fig.~\ref{fig:bisectorij} for an illustration. 

The second part of the statement follows because visibility from $p_i$ is one of the conditions of events of type (iii).
\end{proof}

To prove our bounds, we also use this well-known property of visibility in 1.5D terrains, known as \emph{order claim}:

\begin{lemma}[Claim~2.1 in \cite{ben2007constant}] \label{lem:order}
Let $a,b,c,$ and $d$ be four points on $\T$ such that
$x(a)<x(b)<x(c)<x(d)$. If $a$ sees $c$ and $b$ sees $d$, then $a$
sees $d$.
\end{lemma}

We denote by  $\vorvis (\T,\A_{\ell})$ and $\colvis (\T,\A_{\ell})$ the Voronoi and colored visibility maps of $\T$ assuming that viewpoints can only see themselves and to their left. Further, we denote by \vorviewshed{p}{\A_{\ell}} the Voronoi viewshed of $p$ under the same assumption. $\vorvis (\T,\A_{r})$, $\colvis (\T,\A_{r})$ and \vorviewshed{p}{\A_{r}} are defined analogously using visibility to the right.
We can now prove the following:

\begin{figure*}[t]
    \centering
    \includegraphics[scale=1]{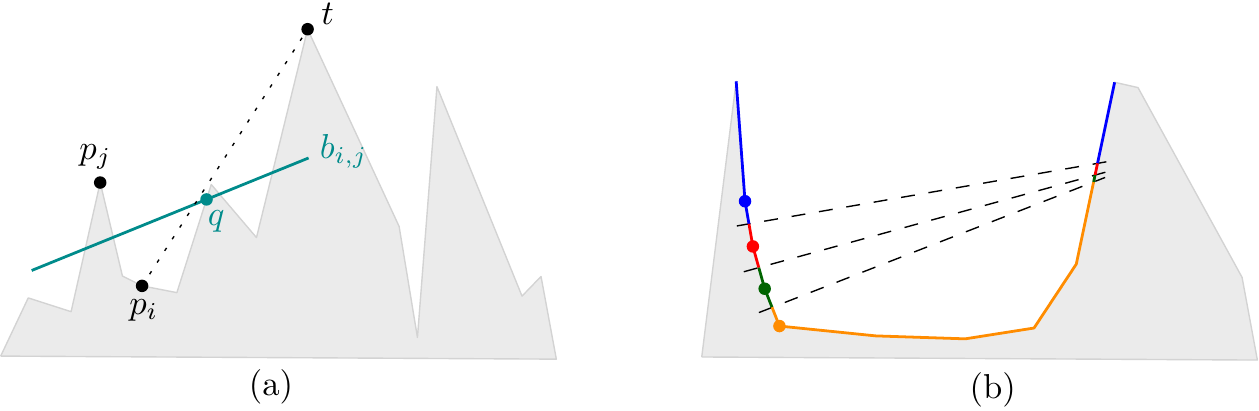}
    \caption{(a) 
    Illustration of the charging scheme of events of $\vorvis (\T,\A_{r})$ at the intersection of a bisector and $\T$ (proof of Theorem~\ref{thm:complexity}): The event $q$ is charged to $p_i$ because no point $t$ to the right of $q$ belongs to $\vorviewshed{p_i}{\A_{r}}$.
    (b) An instance where $k_v= k_c+2m-2$: $\colvis (\T,\A)$ consists of three portions (a portion visible by all viewpoints surrounded by two portions not visible by any viewpoint), while in $\vorvis (\T,\A)$ (illustrated in the figure) the visible portion is subdivided into $2m-1$ parts.}
    \label{fig:lower-bound}
\end{figure*}

\begin{theorem} \label{thm:complexity}
Given a terrain $\T$ with $n$ vertices and a set $\A$ of $m$ viewpoints placed on vertices of $\T$, the following bound holds:
\[
  k_v\leq \min\{k_c+m^2,2k_c+8m-4\}.
\]
\end{theorem}

\begin{proof}
Since the vertices of $\T$ are counted in both $k_v$ and $k_c$, we exclude them from our analysis.

We start by proving that $k_v\leq k_c+m^2$. Notice that events of type (i) and (ii) are also events of $\colvis (\T,\A)$. Let us prove that there are at most $m^2$ events of type (iii).

Let $p_i,p_j$ be a pair of viewpoints. If $p_i$ and $p_j$ are at the same height, $b_{i,j}$ is vertical and only intersects $\T$ once, so there is at most one event of $\vorvis (\T,\A)$ on $b_{i,j}\cap \T$. Otherwise, we assume without loss of generality that $p_i$ is lower than $p_j$. By Lemma~\ref{lem:key} the only candidates for events $q_{i,j}$ or $q_{j,i}$ of type (iii) are the left-most intersection point of type $b_{i,j}\cap \T$ among all such points to the right of $p_i$ and the right-most one among all points to the left. Thus, every pair of viewpoints creates at most two events of type (iii).

We next prove the second upper bound for $k_v$. We denote by $k_v^{\ell}$, $k_c^{\ell}$, $k_v^{r}$ and $k_c^{r}$ the total complexity of all the regions of $\vorvis (\T,\A_{\ell})$, $\colvis (\T,\A_{\ell})$, $\vorvis (\T,\A_{r})$ and $\colvis (\T,\A_{r})$, respectively.

Each event of $\colvis (\T,\A)$ can be uniquely assigned to an event of either $\colvis (\T,\A_{\ell})$ or $\colvis (\T,\A_{r})$: If the event concerns viewpoint $p_i$ (becoming visible or invisible) and it is to the left of $p_i$, the same event appears in $\colvis (\T,\A_{\ell})$ and it is assigned to it. If the event is to the right of $p_i$, it is assigned to the same event in $\colvis (\T,\A_{r})$. If the event is on $p_i$, it is easy to see that $p_i$ is either the left-most point of $\T$, in which case we assign it to the same event in $\colvis (\T,\A_{r})$, or the right-most point of $\T$, in which case we assign it to the same event in $\colvis (\T,\A_{\ell})$.

Each event of $\colvis (\T,\A_{\ell})$ or $\colvis (\T,\A_{r})$ that did not get any event of $\colvis (\T,\A)$ assigned to it lies  at the same position of some viewpoint that is  not the left-most or the right-most point of $\T$. Indeed, if $p_i\in \A$ is such a viewpoint, then, at the position where it lies, $p_i$ becomes visible  in $\colvis (\T,\A_{r})$ and invisible in $\colvis (\T,\A_{\ell})$\footnote{Strictly speaking, in $\colvis (\T,\A_{\ell})$ $p_i$ becomes invisible immediately to its right.}, but there are no such events in $\colvis (\T,\A)$ (where there is a portion of $\T$ visible from $p_i$ containing $p_i$ not on its boundary but in its interior). This proves that $k_c^{\ell}+k_c^r\leq k_c+2m$. 

Next, we show a relationship between $k_v^r$ and $k_c^r$. Suppose that we traverse $\vorvis (\T,\A_{r})$ from left to right, and we stop at every event that is not an event of $\colvis (\T,\A_{r})$, that is, the event is at the intersection of a bisector $b_{i,j}$ and $\T$. Since in $\vorvis (\T,\A_{r})$ viewpoints can only see themselves and to their right, $p_i$ and $p_j$ are to the left of the event $q$. Without loss of generality, suppose that $p_j$ is to the left of $p_i$. As in the proof of Lemma~\ref{lem:key}, if $p_i$ was higher than $p_j$, no point on $b_{i,j}$ to the right of $p_i$ would be visible from $p_j$, contradicting the existence of the event at the intersection of $b_{i,j}$ and $\T$. Further, $p_i$ and $p_j$  are not at the same height because both are to the left of $q$. Hence, $p_i$ is lower than $p_j$. Let $t$ be a point to the right of $q$ that is visible from $p_i$ (see Fig.~\ref{fig:lower-bound}a). By the Lemma~\ref{lem:key} with $a=p_j$, $b=p_i$, $c=q$ and $d=t$, $t$ is also visible from $p_j$. By Lemma~\ref{lem:key}, $t$ is closer to $p_j$ than to $p_i$. Therefore, $t\notin \vorviewshed{p_i}{\A_{r}}$. This implies that there is no portion of \vorviewshed{p_i}{\A_{r}} to the right of $q$ and, in particular no more event caused by the intersection of a bisector of $p_i$ and another viewpoint. We charge the event $q$ to $p_i$, and obtain that there are at most $m-1$ events of this type. Hence, $k_v^r\leq k_c^r + m-1$.


Finally, we derive a bound for $k_v$ based on $k_v^r$ and $k_v^{\ell}$. 

Let us take a continuous portion $\T'$ of \T that belongs to the Voronoi viewshed of some viewpoint $p_i$ in $\vorvis (\T,\A_{r})$, and to the Voronoi viewshed of some viewpoint $p_j$ in $\vorvis (\T,\A_{\ell})$. Let $\T'$ be maximal with this property. Notice that $p_i$ is to the left of $\T'$, while $p_j$ is to its right. Furthermore, in $\vorvis (\T,\A$), every point of $\T'$ belongs to the Voronoi viewshed of $p_i$ or $p_j$. We next show that $b_{i,j}$ intersects $\T'$ at most once. The claim is clear when $y(p_i)=y(p_j)$. Otherwise, we assume without loss of generality that $p_i$ is lower than $p_j$. Let $q$ be the left-most intersection point (if any) between $b_{i,j}$ and $\T'$. We have that $q$ is to the right of $p_i$. Additionally, all points of $\T'$ to the right of $q$ are visible from $p_i$ because they belong to the Voronoi viewshed of $p_i$ in $\vorvis (\T,\A_{r})$. By Lemma~\ref{lem:key}, all points of $\T'$ to the right of $q$ are closer to $p_j$ than to $p_i$. In consequence, there is no intersection point between $b_{i,j}$ and $\T'$ to the right of $q$, and $b_{i,j}$ intersects $\T'$ at most once. This implies that $\T'$ gets split into at most two portions of the final diagram. 

The situation where in at least one of $\vorvis (\T,\A_{r})$ or $\vorvis (\T,\A_{\ell})$ a portion does not have any visible viewpoint is trivial. 

Consequently, $k_v\leq 2(k_v^r + k_v^{\ell})$.

Putting everything together,
\begin{align*}
k_v ~&\le~ 2(k_v^r + k_v^{\ell})\\
        &\le~ 2(k_c^r+k_c^{\ell}+2m-2)\\
        &\le~2(k_c+4m-2)=2k_c+8m-4.
        \qedhere
\end{align*}
\end{proof}

Regarding lower bounds, we show the following:

\begin{example}
There exists a terrain with $n$ vertices and a set of $m$ viewpoints placed on vertices of the terrain such that $k_v= k_c+2m-2$. The construction is illustrated in Fig.~\ref{fig:lower-bound}b.
\end{example}

\section{Computation of the Voronoi visibility map} \label{sec:VorVis}


The algorithm we propose is simple: We sweep the terrain from left to right, and maintain the set of visible points in a balanced binary search tree, where the key of every viewpoint is the Euclidean distance to the point of the terrain currently swept by the sweep line; the relevant viewpoint is always the closest visible one. The algorithm is based on the observation that maintaining the whole set of viewpoints sorted by distance to $\T$ might be expensive (since a bisector of two viewpoints might intersect the terrain $\Theta(n)$ times), while, by Lemma~\ref{lem:key}, maintaining the set of the \emph{visible} ones is not (since, out of the potential $\Theta(n)$ intersections, the two viewpoints are visible in at most two). Thanks to this observation, new events of $\vorvis (\T,\A)$ are found in $O(\log m)$ time rather than $O(m)$. We next present the details.

The algorithm sweeps the terrain from left to right and stops at points that are candidates for event points. The candidates for events of type (i) and (ii) are the events of $\colvis (\T,\A)$. We explain in Section~\ref{subsec:cand} which are the candidates for events of type (iii).

\subsection{Events of $\colvis (\T,\A)$} 

We compute $\colvis (\T,\A)$ using the version of the algorithm from~\cite{ter-vis2014} that returns a doubly-linked list with the vertices of $\colvis (\T,\A)$ sorted from left to right, together with the visibility information provided as follows: The visible viewpoints are specified for the first component of $\colvis (\T,\A)$ and, for the other components, the algorithm outputs the changes in the set of visible viewpoints with respect to the component immediately to the left.

\subsection{Candidates for events of type (iii)} \label{subsec:cand}

We next describe the candidates for events of type (iii) associated with a pair of viewpoints $p_i,p_j\in \A$. 

If $p_i$ and $p_j$ are at the same height, the only intersection point of $b_{i,j}$ with $\T$ lies between both viewpoints. If such point is visible from both $p_i$ and $p_j$, we add it as a candidate for event of type (iii). 

Otherwise, we may assume, without loss of generality, that $p_i$ is lower than $p_j$. By Lemma~\ref{lem:key} the only candidates for events of type (iii) involving $p_i$ and $p_j$ are the left-most intersection point of type $b_{i,j}\cap \T$ among all such points to the right of $p_i$ and the right-most one among all points to the left. For the sake of simplicity, we first assume that $b_{i,j}$ is not tangent to $\T$ at any of these intersection points. Then each of these intersection points is added to the list of candidates for events swept by the line if and only if it is visible from both $p_i$ and $p_j$. 

Finally, let $q$ be one of the two candidates for events of type (iii) involving $p_i$ and $p_j$. Suppose that $q$ is to the right of $p_i$ (the other case is symmetric). If $b_{i,j}$ is tangent to $\T$ at $q$, points of $\T$ infinitesimally to the left or right of $q$ are closer to $p_i$ than to $p_j$ (while $q$ is equidistant). Additionally, $p_i$ becomes invisible right after $q$. In consequence, it is not needed to add $q$ to the list of candidates for events of type (iii): Right before $q$, the algorithm knows that $p_i$ is closer to the terrain than $p_j$. At $q$, the algorithm processes that $p_i$ becomes invisible, and $p_j$ (if it is visible) automatically gets higher priority than $p_i$ in the list of candidates for the ``owner" of the current Voronoi visibility region.
The key argument (there are no more candidates for events of type (iii) to the right of $q$) also holds in this case.

\begin{figure}
	\centering
 \lineskip=-\fboxrule
  \ovalbox{\begin{minipage}{\dimexpr \textwidth-55\fboxsep-1\fboxrule}
\begin{algorithmic}[1]
\vspace{0.25cm}
\Require  $\T,\A,E$
\Ensure  $\vorvis(\T,\A)$
\State  $H:= \emptyset$, $t_{\ell}:=$ left-most point of $\T$, and $p_{*} := \bot$
\While{$E \ne \emptyset$ }
\State extract the next element $q$ of $E$
\If {$q$ is the last element of $E$}
 \State output $((t_{\ell},q),p_{*})$
 \State break
\ElsIf{some viewpoint $v$ becomes visible at $q$}
    \State insert $v$ in $H$ 

\ElsIf{some viewpoint $v$ becomes invisible at $q$ }
\State delete $v$ from $H$ 
    
    
 \ElsIf{$q$ is an intersection point between $\T$ and $b_{i,j}$ }
\State update the positions of $p_i$ and $p_j$ in $H$ 
    
    \EndIf
    \State update $p_{min}$
    \If {$p_{\min}\neq p_{*}$}
    \State output $((t_{\ell},q),p_{*})$
    \State   $t_{\ell}:= q$, $p_{*}:=p_{\min}$
    \EndIf

\EndWhile
\vspace{0.20cm}

\end{algorithmic}
 \end{minipage}}
\caption{Computation of  $\vorvis(\T,\A)$. $E$ is the list of potential events, $H$ is the tree containing the viewpoints that are currently visible, $t_{\ell}$ is the left endpoint of the current portion of $\T$, $p_{*}$ is the closest visible viewpoint in that portion, and $p_{\min}$ is the viewpoint in $H$ with the minimum key.}\label{alg:cap}
\label{alg:vorvis}
\end{figure}

\subsection{Data structures} \label{subsec:datastr}

The algorithm uses the following data structures. 

We maintain a balanced binary search tree $H$ that contains the viewpoints that are visible at the current point of the sweep. These viewpoints are sorted in the tree according to their corresponding key, which is the distance from the viewpoint to the current intersection point between the sweep line and the terrain. The keys are not stored in the tree because they change as the sweep line moves, but each of them can be computed when needed in constant time. The algorithm always chooses as the ``owner'' of the current Voronoi visibility region the viewpoint of $H$ with the minimum key.

In $H$, we perform insertions and deletions when viewpoints become visible and invisible, respectively. During these operations, when at some node of the tree we need to decide whether we move to its left or right subtree, we simply compute the key associated to the viewpoint in that node, and compare it with the key of the viewpoint that we want to insert or delete. Therefore, insertions and deletions can be performed in the standard way in $O(\log m)$ time.

When the sweep line encounters a candidate for an event of type (iii) (let us call it $q$), the relative order of two visible viewpoints with respect to their current distance to the terrain changes (formally speaking, it changes right after $q$). A possible way to reflect this in $H$ is to delete from the tree one of the two viewpoints associated with $q$, and then insert it again using as keys the distances from the viewpoints to a point of $\T$ infinitesimally to the right of $q$ (and still to the left of the next event in the list). Thus, candidates for events of type (iii) can be processed in $H$ in $O(\log m)$ time.

Additionally, we use a data structure that allows us to answer ray-shooting queries in $\T$ in $O(\log n)$ time~\cite{cegghss-rspugt-94}. Such queries are used to decide whether a given pair of points are mutually visible, and to find the relevant intersections between $\T$ and the bisector of a pair of viewpoints.

\subsection{Description of the algorithm} 

Given $q,r$ on $\T$ with $x(q)<x(r)$, we denote by $\T(q,r)$ and $\T[q,r]$ the open and closed portion of the terrain between $q$ and $r$, respectively.

Our algorithm, outlined in Fig.~\ref{alg:vorvis}, takes as input $\T$, $\A$ and a list $E$ of potential events sorted from left to right containing all events of $\colvis (\T,\A)$ together with the $O(m^2)$ candidates for events of type (iii). The list $E$ also contains an event at the right-most point of the terrain.

The algorithm outputs $\vorvis (\T,\A)$ as a list of pairs $((q,r),p_i)$ such that $p_i$ is the closest visible viewpoint in $\T(q,r)$ (if $\T(q,r)$ is not visible from any viewpoint, we output $((q,r),\bot)$). The variables $t_{\ell}$ and $p_{*}$ in the algorithm refer to the left endpoint of the portion of $\T$ currently analyzed by the algorithm and the closest visible viewpoint in that portion, respectively. 
The variable $p_{\min}$ refers to the viewpoint in $H$ with the minimum key (if $H$ is empty, $p_{\min}=\bot$).

Initially, $H:= \emptyset$, $t_{\ell}:=$ left-most point of $\T$, and $p_{*} := \bot$.

We repeat the following procedure until $E$ is empty: We  extract the next element $q$ from $E$, and proceed according to four cases, corresponding to lines 4, 7, 9, and 11 of the pseudocode in Fig.~\ref{alg:vorvis}. For the sake of simplicity, in the description in Fig.~\ref{alg:vorvis} we deliberately ignore the situation where several events of distinct type occur at the same point of $\T$, which we tackle in the next paragraph. The cases in lines 4, 7 and 9 are clear. Regarding the case starting at line 11, in line 12 we update the positions of $p_i$ and $p_j$ in $H$ as explained in Section~\ref{subsec:datastr} (see the paragraph where we discuss the case where the sweep line encounters a candidate for an event of type (iii)). We also point out that, if $q$ is an intersection point between $\T$ and more than one bisector of type $b_{i,j}$, the bisectors can be processed in any order.\footnote{By our general position assumptions, $q$ is not equidistant from three or more viewpoints, so at most one of the bisectors through $q$ might involve $p_{*}$.}








It remains to explain how to deal with the situation where several events of distinct type occur at the same point of $\T$. In this case, we first perform the modifications in $H$ triggered by \emph{all} the events at that point (insertions of viewpoints becoming visible, deletions of viewpoints becoming invisible and updates of the positions of pairs of viewpoints). After updating $H$ in this way, we update $p_{\min}$; if $p_{\min}\neq p_{*}$, we output $((t_{\ell},q),p_{*})$, set $t_{\ell}:= q$, and set $p_{*}:=p_{\min}$.



\subsection{Correctness and running time} 

We first show that the algorithm for $\vorvis (\T,\A)$ always selects the closest visible viewpoint. Changes in the visibility status of the viewpoints correspond to events of $\colvis (\T,\A)$, which are added to $E$, so the set of visible viewpoints contained in $H$ is correct at any time of the sweep. Regarding the distances from the viewpoints to the terrain, every time that a viewpoint is swept or becomes visible, it is inserted in $H$ correctly (according to its current distance to the terrain).
Changes in the order of the visible viewpoints with respect to their distances to $\T$
coincide with intersections of $\T$ with the bisectors among them. As argued in the proof of Theorem~\ref{thm:complexity}, for every pair of viewpoints it happens at most twice that both viewpoints are visible at an intersection point between $\T$ and their bisector. Such an event is precomputed and stored in $E$, and later processed by the algorithm. 

We next analyze the complexity of the algorithm.

The map $\colvis (\T,\A)$ can be computed in $O(n + (m^2 + k_c) \log n)$ time using the algorithm in~\cite{ter-vis2014}. This map has at most $k_c$ regions; however, due to the fact that several viewpoints might become visible or invisible at the same time, when sweeping $\colvis (\T,\A)$ from left to right, the number of times that a viewpoint becomes visible or invisible, added over all viewpoints, can be higher; an upper bound of $k_c+m^2$ is given in~\cite{ter-vis2014}. Each time that a viewpoint changes its visibility status, we perform an insertion or a deletion in $H$, which takes $O(\log m)$ time. The algorithm processes at most $m^2$ intersections between the terrain and bisectors of endpoints in $O(\log m)$ time each. Consequently, $\vorvis (\T,\A)$ can be extracted from $\colvis (\T,\A)$ in $O((m^2 + k_c) \log m)$ time. The space complexity of the algorithm is the space required to store the terrain, the events and the data structures, that is, $O(n + m^2 + k_c)$. 

We conclude with the following:

\begin{theorem}
The Voronoi visibility map of a 1.5D terrain can be constructed in $O(n + (m^2 + k_c) \log n)$ time and $O(n + m^2 + k_c)$ space.
\end{theorem}

\section{Extensions}
In this section, we present adaptations of the previous algorithm to compute related maps.

\subsection{Higher order Voronoi visibility maps}

We define the $k$th-order Voronoi visibility map $\kvorvis (\T,\A)$ as a partitioning of $\T$ into regions that have the same set of $\ell$ closest visible viewpoints, where $\ell$ is the minimum of $k$ and the number of visible viewpoints in the region. Observe that the $m$th-order Voronoi visibility map is equal to $\colvis(\T,\A)$.  

We can easily compute $\kvorvis (\T,\A)$ by adapting the algorithm from Section~\ref{sec:VorVis}. In this case, we need to maintain two additional variables: the total number $b$ of viewpoints that are visible at the point currently swept by the line, and, from the current set of $\ell$ closest visible viewpoints, the furthest one, denoted $p_{\max}$. Analogously to the algorithm for $\colvis (\T,\A)$, for space reasons our algorithm for $\kvorvis (\T,\A)$ returns a doubly-linked list with the vertices of $\kvorvis (\T,\A)$ sorted from left to right, together with the following information: The set of $\ell$ closest visible viewpoints is specified for the first component of $\kvorvis (\T,\A)$ and, for the other components, the algorithm outputs the changes in the set of $\ell$ closest visible viewpoints with respect to the component immediately to the left.

Let $q$ be the next element  from the list of events $E$, computed as in the previous section. We explain in detail the case where one or more viewpoints become visible at $q$, and leave the remaining cases to the interested reader. Let $\A'$ denote the set of viewpoints becoming visible at $q$. We update $b$. If, after this update, $b\leq k$, we report vertex $q$ together with the set $\A'$ (containing the new viewpoints in the set of $\ell$ closest visible viewpoints). We also insert the viewpoints of $\A'$ in $H$. Otherwise, let $b'$ and $b$ be the number of visible viewpoints right before $q$ and at $q$, respectively. If $b'<k$, we remove from $\A'$ the set of $k-b'$ closest viewpoints to $q$ (obtained after sorting the viewpoints of $\A'$ according to their distance to $q$), we add these viewpoints to a set $\A'_{in}$, and we insert them in $H$. After possibly performing this operation in $\A'$, we proceed as follows: We extract the closest viewpoint to $q$ of $\A'$; if it is closer to $q$ than $p_{\max}$, we add this viewpoint to $\A'_{in}$, we insert it in $H$, we add viewpoint $p_{\max}$ to $\A_{out}$, and we update $p_{\max}$. Notice that $p_{\max}$ can be updated by finding the predecessor in $H$ of the ``old'' $p_{\max}$, that is, in $O(\log m$) time. We repeat this process until $\A'$ is empty or the next element in $\A'$ is farther to $q$ than $p_{\max}$. Then we insert the remaining viewpoints of $\A'$ (if any) in $H$. Finally, we report vertex $q$ together with the set $\A'_{in}$ (containing the new viewpoints in the set of $\ell$ closest visible viewpoints) and the set $\A_{out}$ (containing the viewpoints that stop belonging to the set of $\ell$ closest visible viewpoints).

Clearly, every change in the visibility status of a viewpoint and every intersection of $\T$ with the bisector of two visible viewpoints can be processed in $O(\log m + \log n)$ time. Hence, we obtain:

\begin{theorem}
The $k$th-order Voronoi visibility map of a 1.5D terrain can be constructed in $O(n + (m^2 + k_c) \log n)$ time and $O(n + m^2 + k_c)$ space.
\end{theorem}

\subsection{Other distances}

Given $q,r$ on $\T$ with $x(q)<x(r)$, two other natural distances between $q$ and $r$ are the Euclidean length of the portion $\T[q,r]$, which we will call \emph{Euclidean distance along the terrain}, and the number of vertices in the portion $\T(q,r)$, which we will call \emph{link distance}.\footnote{For the link distance, we take the open portion of the terrain $\T(q,r)$ so that any two points on the same edge (including the endpoints) are at (link) distance zero.}
We may define the Voronoi visibility map of $\T$ based on these distances.

The relevant difference with respect to the standard case is the shape of the bisectors between two viewpoints $p_i$ and $p_j$. In the case of the Euclidean distance along the terrain, there is exactly one point of $\T$ that is equidistant to $p_i$ and $p_j$, and this point can be computed in $O(\log n)$ time after preprocessing $\T$ so that the Euclidean distance along the terrain between any pair of vertices of $\T$ can be computed in $O(1)$ time.\footnote{If we store, for every vertex $q$ of $\T$, the Euclidean distance along the terrain $q_d$ between $q$ and the left-most point of $\T$, then the Euclidean distance along the terrain between vertices $q,r$ of $\T$ such that $x(q)<x(r)$ is $r_d-q_d$.} Regarding the link distance, if there is an odd number of vertices between $p_i$ and $p_j$, there is exactly one vertex of $\T$ that is equidistant to $p_i$ and $p_j$, and this vertex can be computed in $O(1)$ time. However, if there is an even number of vertices between $p_i$ and $p_j$, there is an open edge of $\T$ such that all of its points are at the same link distance from $p_i$ and $p_j$. In this case, we must either allow the border between two consecutive Voronoi regions to be 1-dimensional, or, if simplicity is more desirable, we might (artificially) select an interior point of this edge as the intersection point between $\T$ and the bisector of $p_i$ and $p_j$. 

After adding the corresponding candidates for events of type (iii) based on the explanations in the previous paragraph, the rest of the algorithm is equal to the one for the general case. The running time remains the same because, given a pair of points on $\T$, in both cases the distance between them can be computed in $O(1)$ time. Therefore, we conclude:

\begin{theorem}
The Voronoi visibility map of a 1.5D terrain with respect to the Euclidean distance along the terrain or to the link distance can be constructed in $O(n + (m^2 + k_c) \log n)$ time and $O(n + m^2 + k_c)$ space.
\end{theorem}

\section{Computation of $r^*$} \label{sec:minr}

We recall that $r^*$ is the minimum value of $r$ such that, if the viewpoints can only see objects that are within distance $r$, the visibility map of $\T$ does not change.

 Let $\A^{r}$ denote the set of viewpoints $\A$ with the restriction that the visibility range of the viewpoints is $r$. We then may define $\vis(\T,\A^{r}$), $\vorvis(\T,\A^{r})$\ldots in the natural way. Notice that, for $\A^{\infty}$, we obtain the same objects as in the standard case.
 
Let $d(x,y)$ denote the Euclidean distance between two points $x,y\in \mathbb{R}^2$. 
 	
\begin{lemma} \label{lem:minr}
$ r^*=\underset{i=1,\ldots, m}{\max}\{    \underset{x \in \vorviewshed{p_i}{\A^{\infty}}}{\sup} ~ d(p_i,x)\}$.
\end{lemma}

\begin{proof}	
Let $p_i$ and $x$ be a viewpoint and a point of $\T$ achieving the maximum in the right hand expression. If $r^*<d(p_i,x)$, $x$ would not be visible from $p_i$ in $\vis(\T,\A^{r^*})$. Since $x$ belongs to the boundary of $\vorviewshed{p_i}{\A^{\infty}}$, all other viewpoints seeing $x$ have a distance to $x$ that is greater than or equal to $d(p_i,x)$; thus, $x$ would also not be visible from any of them in $\vis(\T,\A^{r^*})$. Since $x$ is visible in $\vis(\T,\A^{\infty})$\footnote{It follows from our definition of visibility that the maximal visible portions of $\T$ are closed and, hence, the points on the boundary of the Voronoi viewsheds are visible.}, we reach a contradiction. Therefore, $r^*\geq d(p_i,x)$.

On the other hand, to keep $\vis(\T,\A^{\infty})$ unchanged, it is enough to maintain the closure of $\vorviewshed{p_i}{\A^{\infty}}$ visible for all $i$, since $\vis(\T,\A^{\infty})$ is equal to the union of the closures of the regions $\vorviewshed{p_i}{\A^{\infty}}$. If we set a visibility range of $\underset{x \in \vorviewshed{p_i}{\A^{\infty}}}{\sup} ~ d(p_i,x)$, the closure of $\vorviewshed{p_i}{\A^{\infty}}$ indeed remains visible. Consequently, $r^*\leq \underset{i=1,\ldots, m}{\max}\{    \underset{x \in \vorviewshed{p_i}{\A^{\infty}}}{\sup} ~ d(p_i,x)\}$.
\end{proof}  

Using this characterization of $r^*$, we can prove the following:

\begin{theorem}
The problem of computing the minimum value $r^*$ such that $\vis(\T,\A^{r^*}) = \vis(\T,\A^{\infty})$ can be solved in $O(n+(m^2+k_c) \log n)$ time.
\end{theorem}

\begin{proof} 
By Lemma~\ref{lem:minr}, it suffices to consider the distances between the vertices of $\vorvis(\T,\A^{\infty})$ (that is, the points on the boundary of the Voronoi viewsheds) and their associated viewpoints. Consequently, the problem can be trivially solved in linear time if $\vorvis(\T,\A^{\infty})$ is known.
\end{proof} 

\section{Final remark}

As indicated in~\cite{ter-vis2014}, in the running time of the algorithm to compute $\colvis(\T,\A)$, the term $m^2\log n$ disappears if we assume that no two viewpoints change from invisible to visible at the same point of $\T$. This can always be achieved by infinitesimally perturbing the terrain. However, such a perturbation does not make the same term disappear from the running time of the presented algorithm to compute $\vorvis(\T,\A)$. Given that one of the bounds in Theorem~\ref{thm:complexity} guarantees that $k_v=O(k_c+m)$, it remains as an open problem to design an algorithm for $\vorvis(\T,\A)$ that is equally faster than that for $\colvis(\T,\A)$ for all possible instances.




\bibliographystyle{elsarticle-harv}
\bibliography{refs}

\end{document}